\documentclass[11pt]{article}

\usepackage[cmex10]{amsmath}
\usepackage{amsfonts}
\usepackage{amsthm}
\usepackage{amssymb}

\let\sqcup=\barwedge
\let\amalg=\doublebarwedge

\usepackage{url}

\newcommand{\ud}{\, \mathrm{d}}

\theoremstyle{definition} \newtheorem{theorem}{Theorem}[section]
\theoremstyle{definition} 
\theoremstyle{definition} 
\theoremstyle{definition} \newtheorem{corollary}[theorem]{Corollary}
\theoremstyle{definition} 

\newcommand{\eat}[1]{}

\begin{document}

\title{Degradation Analysis of\\ Probabilistic Parallel Choice Systems}
\author{Avinash Saxena\footnote{avinashsxna@iitkgp.ac.in} \ and Shrisha
  Rao\footnote{shrao@ieee.org}}
\maketitle

\begin{abstract}
Degradation analysis is used to analyze the useful lifetimes of
systems, their failure rates, and various other system parameters like
mean time to failure (MTTF), mean time between failures (MTBF), and
the system failure rate (SFR).  In many systems, certain possible
parallel paths of execution that have greater chances of success are
preferred over others.  Thus we introduce here the concept of
\emph{probabilistic parallel choice}.  We use binary and $n$-ary
probabilistic choice operators in describing the selections of
parallel paths.  These binary and $n$-ary probabilistic choice
operators are considered so as to represent the complete system
(described as a series-parallel system) in terms of the probabilities
of selection of parallel paths and their relevant parameters.  Our
approach allows us to derive new and generalized formulae for system
parameters like MTTF, MTBF, and SFR.  We use a generalized exponential
distribution, allowing distinct installation times for individual
components, and use this model to derive expressions for such system
parameters. \\

\end{abstract}

\vspace{0.15in}

\section*{Notation}

\begin{tabular}{@{}rp{9cm}@{}}
$X_i$ \,: & Component with index $i$ \\
$P(X_i)$ \,: & Probability that component $X_i$ works successfully \\
$P_{\delta}(X_i)$ \,: & Probability that component $X_i$ fails \\
$R_k$ \,: & Parallel path $k$ \\
$P(R_k)$ \,: & Probability that parallel path $R_k$ works
successfully \\
$P_{\delta}(R_k)$ \,: & Probability that parallel path $R_k$ fails \\
$S$ \,: & A complete system \\
$P(S,t \geq T)$ \,: & Probability of success of system $S$ for \(t
\geq T\) \\
$P_{\delta}(S,t \geq T)$ \,: & Probability of failure of system $S$
for $t \geq T$ \\
$\oplus$ \,: & Non-probabilistic choice operator \\
$\otimes$ \,: & Sequential choice operator \\
$\sqcup_{\psi}$ \,: & Probabilistic choice operator (binary case) \\
$\amalg_{\psi_k}$ \,: & $k$th Probabilistic choice operator
($n$-ary case) \\
$\psi_k$ \,: & Probability with which parallel path $R_k$ is
chosen \\
$\hat t_{0i}$ \,: & Installation time of component $X_i$ \\
$t_{\infty}$ \,: & Time at which probability of successful working
of a system is almost zero \\
$\rho$ \,: & Threshold value of probability of success of system
for system to be reliable in determination of mean time to failure \\
$\lambda_i$ \,: & Failure rate of component $X_i$ which determines
how fast the probability of success decays with time \\
$\lambda_{eq}$ \,: & The equivalent system failure rate given
components with $\lambda_i$ individual failure rates \\
\end{tabular}

\section{Introduction}

Degradation analysis is used to determine system parameters related to
reliability, and is commonly done using reliability block diagrams
(RBDs), assuming time-dependent distributions like exponential,
Weibull, normal, etc.  We propose a model to compute more generalized
formulae for these parameters, and indicate such generalized formulae
for the exponential distribution in particular.  Traditional modelling
techniques such as RBDs treat each component as a node in a diagram;
several such nodes combine in series or parallel combinations to form
a system.  Each node is associated with a time-dependent probability
distribution.  The probabilities of success and failure of these
components are assumed to vary according to the distribution
considered.  Further, the system parameters are calculated based on
the distribution considered.

A major accomplishment of this work lies in development of new and
better formulae for mean time to failure (MTTF), system failure rate
(SFR), mean time between failures (MTBF), reliability time estimation
(RTE), probability density function (pdf), and mean time to repair
(MTTR), using the new approach of probabilistic parallel choice.  Like
the series-parallel systems considered in classical works such as
Bazovsky~\cite{book1} and Birolini~\cite{book2}, we assume each complex
system to be composed of series and parallel combinations of
components.  We use probabilistic choice operators to represent a
system having components in parallel.  The probabilistic choice in
selection of parallel paths extends the traditional model for systems,
as the probabilities of selection of parallel paths are not equally
likely (as is implicitly the case in all prior works).

There are four operators in our model.  These include probabilistic
choice operators for both binary and $n$-ary cases (compare with
Andova~\cite{andova}).  Two other operators, the non-probabilistic
choice operator and the sequential operator, are also used.  Classical
RBDs use a non-probabilistic choice operator in considering selection
among parallel paths, so that all the parallel paths are assumed to be
identical and each path’s selection is considered equally likely.
This assumption in RBDs and related models is certainly a limitation,
as some systems may choose certain paths with higher probabilities
than others, e.g., when a system is designed so that a more
reliable---or less expensive---path shall be chosen prior to, or more
often than, another.

The relevance of probabilistic parallelism thus lies in the fact that
parallel paths are not identical in many real systems.  There can be
some associated advantages and disadvantages in traversing particular
paths: a particular path may be faster but at the same time costlier,
or it may be safer or less disruptive than another path.

A binary probabilistic choice operator as in Andova~\cite{andova}
applies in case of two paths in parallel where there is a
probabilistic choice attached to the selection of each parallel path.
An $n$-ary probabilistic choice operator, also as in
Andova~\cite{andova} applies with $n$ paths in parallel, if there
exists a probabilistic choice in selection of each of the parallel
paths.  Later, in Section~\ref{quantify} we present a method to
quantify the probability $\psi_k$ of selection of a given path $R_k$.
The probabilistic choice operators are used to represent two or more
paths in parallel, with distinct probabilities of selection attached
to all paths.  Unlike in RBDs, each path in parallel is treated
differently, and is given a priority based on the chances of success
associated with it.  This allows us to come up with a better
degradation analysis of systems.  The series system is similar to that
described in Gottumukkala~\cite{GottumukkalaNR}, Bazovsky~\cite{book1}
and Birolini~\cite{book2}, and uses a sequential operator to represent
consecutive components or subsystems in series.  In
Section~\ref{PATFT}, we introduce time-dependent probabilities of
success and failure using the exponential distribution.  In our
analyses, each component is assumed to have a probability of success
and failure that varies exponentially with time.

It is well known that when a system breaks down, not all components of
the system are considered as faulty.  Replacement or repair is
generally considered essential only for those components that have
actually ceased to work properly, and have contributed to the current
breakdown of the system.  Thus, assuming an identical age for each
component in performing degradation analysis of a system is usally not
a very sound idea---some of the components in the system may have
different ages for which they have been operational, and thus their
actual probabilities of success and failure can be quite different
from that obtained by considering identical ages for all components.
We thus assume that each component has a lifetime after which it is to
be repaired or replaced, and that after it is repaired or replaced,
the component is as good as new.  The installation time is the time at
which the component is inserted \emph{ab initio}, is repaired, or is
replaced.  For a time close to an installation time, the chances of
successful working of the component are high, and as the time
increases the chances of its successful working decay.  Earlier models
assumed installation times for all components to be fixed and
identical~\cite{book1,book2}, but in our model we consider different
times of installation for different components in general.  Thus the
probabilities of success of individual components can be assumed to
follow an exponential distribution for times greater than their
installation times.  Based on this, we derive better formulae in
degradation analysis for system parameters like MTTF, SFR, MTBF, MTTR,
and pdf.  The formulae obtained are more generally applicable compared
to previous ones, as earlier work did not incorporate a probabilistic
choice in selection of parallel paths, or variable installation times.

\section{System Model} \label{concept}

A \emph{system} in our model comprises of a series and parallel
combination of independently-working units called \emph{components}.
These components have attached probabilities of success and failure.
In general, these probabilities may be constant or time-dependent,
depending on the type of distribution followed.  We also assume that
all components in a system are independently-working units of that
system, so the success or failure of any component is independent of
the success or failure of other components.

The probability of success of a system is a measure of its overall
performance.  For each component $X_i$, the probability of successful
working of the component is given by $P(X_i)$ and a probability of
failure of the component given by $P_{\delta}(X_{i})=1-P(X_i)$.  The
success or failure of the complete system is dependent upon the
success and failure of the individual components of which it is
composed.

The system is assumed to be repairable system, that is, replacement or
repair of the components is undertaken if they cease to work causing
failure of the entire system.  Any component so repaired or replaced
is assumed to be as good as a new one.

In our model, we assume that each parallel path $R_k$ is chosen with a
different probability $\psi_k$, and that these selection probabilities
of parallel paths remain even though the probabilities of success and
failure of individual components in the paths may be vary over time.
We further assume (without loss of generality) that each parallel path
in a system contains exactly one component.

\subsection{Related Operators}

Our model has four basic operators.  These operators are used for 
representing series and parallel combinations of components.

\begin{itemize} \label{operators}

\item Non-probabilistic choice operator $(\oplus)$:

  $R_1\oplus R_2$ means that the probability of choosing path $R_1$ is
  equal to the probability of choosing path $R_2$.

\item Sequential operator $(\otimes)$:

  The sequential operator $\otimes$ is used such that \(X_1\otimes X_2\)
  implies that the probability of choosing $X_1$ is equal to the
  probability of choosing $X_2$.

\item Binary probabilistic choice operator $\sqcup_{\psi}$:

  \(R_1\sqcup_{\psi}R_2\) represents a parallel system with parallel
  paths $R_1$ and $R_2$, where the probability of selection of path
  $R_1$ is equal to $\psi$, and the probability of selection of path
  $R_2$ is equal to $1-\psi$, with \(0 \leq \psi \leq 1\).

\item $n$-ary probabilistic choice operator \((\amalg_{\psi_k})\),
  with \(n\geq 2\):

  \(R_1\amalg_{\psi_1} R_2 \amalg_{\psi_2} \ldots \amalg_{\psi_{n-1}}
  R_n\) represents a parallel system with $n$ parallel paths $R_1$
  through $R_n$, where the probability of choosing path $R_k$ is
  $\psi_k$, \(1 \leq k \leq n-1\), and the probability of choosing
  path \(R_n$ is $1-\sum_{k=1}^{n-1}\psi_k\).

\end{itemize}

We now present a basic theorem relating the binary and $n$-ary
probabilistic choice operators.  It is useful in calculating the
overall probability of success or failure of complex systems in terms
of the relevant parameters of their components.

\begin{theorem} \label{th1}
The following hold with respect to $\sqcup_{\psi}$ and $\amalg_{\psi_k}$.

\begin{itemize}

\item[$\mathrm{(a)}$] \(P\left(R_1 \sqcup_{\psi} R_2 \right)= \psi
  P(R_1)+(1-\psi) P(R_2)\)

\item[$\mathrm{(b)}$] \(P\left(R_1 \amalg_{\psi_1} R_2
  \amalg_{\psi_2}R_3\ldots \amalg_{\psi_{n-1}} R_n\right) \\= \psi_1
  P\left( R_1 \right)+\psi_2 P\left( R_2 \right) +\ldots+\psi_{n-1}
  P\left( R_{n-1} \right)\\ + \left( 1 -\sum_{k=1}^{n-1} \psi_k
  \right) P\left( R_n\right)\)
\end{itemize}
\end{theorem}

\begin{proof}
For part (a), we note that \(\left(R_1\sqcup_{\psi} R_2 \right)\)
represents a system which has paths $R_1$ and $R_2$ as alternatives,
with the probability of choosing path $R_1$ as $\psi$ and probability
of choosing path $R_2$ as $1-\psi$.  Thus the probability of success of
the entire system \(\left(R_1 \sqcup_{\psi} R_2 \right)\) can be written
with the basic total probability formula:
\begin{eqnarray} \label{eq23}
P(A)=\sum_{k=1}^{n}P(A|X_k)P(X_k)
\end{eqnarray}

Using~\eqref{eq23} we can compute \(P\left(R_1 \sqcup_{\psi} R_2 \right)\) as:
\begin{equation*}
P \left(R_1 \sqcup_{\psi} R_2 \right) = P\left((R_1 \sqcup_{\psi} R_2) | R_1\right) P(R_1) + P\left((R_1\sqcup_{\psi} R_2) | R_2\right) P(R_2)
\end{equation*}

Here $P(R_1)$ and $P(R_2)$ denote the probabilities of successful
working of paths $R_1$ and $R_2$.

From this, we get:
\begin{equation*}
P \left(R_1 \sqcup_{\psi} R_2\right) = \psi P(R_1)+(1-\psi) P(R_2)
\end{equation*}

Hence part (a) holds.

For part (b), we similarly extend the formula to $n$ components in
parallel, with the probability of selection of the parallel path $R_k$
being as previously given.

Using the total probability formula~\eqref{eq23} we can compute the
value of \(P\left(R_1 \amalg_{\psi_1} R_2 \amalg_{\psi_2} R_3 \ldots
\amalg_{\psi_{n-1}} R_n \right)\) as:
\begin{equation*}
\sum_{k=1}^{n-1}P(\left(R_1 \amalg_{\psi_1} R_2 \amalg_{\psi_2} R_3
  \ldots \amalg_{\psi_{n-1}} R_n \right)|R_k)P(R_k),
\end{equation*}

where $P(R_k)$ denotes the probability of successful working of path
\(R_k$, $\forall k \in{1, 2, \ldots, n}\).  Thus we get:
\begin{eqnarray*}
P\left(R_1 \amalg_{\psi_1} R_2 \amalg_{\psi_2} R_3\ldots \amalg_{\psi_{n-1}} R_n\right) &= \psi_1 P \left( R_1 \right) + \psi_2 
P\left( R_2 \right) +\ldots +\psi_{n-1} P\left( R_{n-1} \right) \\&+ \left( 1 -\sum_{k=1}^{n-1} \psi_k \right) P\left( R_n \right).
\end{eqnarray*}

Hence part (b) holds. 
\end{proof}

In line with our assumption about parallel paths having exactly one
component apiece, an obvious corollary may be stated.

\begin{corollary} \label{th2}
If each parallel path has exactly one component $X_i$ in it, the
following is true of $\sqcup_{\psi}$ and $\amalg_{\psi_i}$.

\begin{itemize}

\item[(a)] \(P\left(X_1 \sqcup_{\psi_i} X_2 \right)= \psi_i P(X_1)+(1-\psi_i) P(X_2)\)

\item[(b)] \(P\left(X_1 \amalg_{\psi_1} X_2 \amalg_{\psi_2}X_3\ldots \amalg_{\psi_{n-1}} X_n\right) \\= \psi_1 P\left( X_1 \right)+\psi_2 P\left( X_2 \right) +\ldots+\psi_{n-1} P\left( X_{n-1} \right)\\ + \left( 1 -\sum_{i=1}^{n-1} \psi_i \right) P\left( X_n\right)\)
\end{itemize}

\end{corollary}

Every system in this model is assumed to be composed of some series
and parallel combination of components.  The two basic types of
systems, that is, series and parallel systems, are described in
Sections~\ref{series} and~\ref{parallel} using operators defined in
Section~\ref{operators}.

\subsection{Series System} \label{series}

An $n$-component series system with $n \geq 2$ components $X_i$ fails
if any one of the $X_i$ fails to work.

\(P(X_1), P(X_2), \ldots, P(X_n)\) denote the probabilities of success
of components \(X_1, X_2, \ldots, X_n\).  Thus the probabilities of
failure of components \(X_1,X_2,\ldots,X_n\) are \(1-P(X_1),
1-P(X_2),\ldots,1-P(X_n)\).

A series system with $n$ components \((X_1, X_2, \ldots, X_n)\) can be
represented as:
\begin{eqnarray*}
X_1\otimes X_2\otimes X_3 \ldots \otimes X_n.
\end{eqnarray*}

As already mentioned, the components of the system are assumed to be
independently working units, and thus the probability of success of
the series system is given by:
\begin{eqnarray*} \label{eq1}
P(X_1\otimes \ldots \otimes X_n) = P(X_1) P(X_2)\ldots P(X_n) 
\end{eqnarray*}

This formula is similar to that given by Yu and Fuh~\cite{Yu}.

The two possible outcomes of the system are either a success or a
failure, which are mutually exclusive and cumulatively exhaustive
events.

Thus, the probability of failure of the series system, as by Yu and
Fuh~\cite{Yu}, is given by:
\begin{eqnarray} \label{eq212}
P_{\delta}(X_1\otimes \ldots \otimes X_n) & = &
1-P(X_1\otimes X_2\otimes \ldots \otimes X_n)
\\
& = & 1-P(X_1)P(X_2)\ldots P(X_n)
\end{eqnarray}

The overall probability of success of the series system is smaller
than the component with the least probability of success, as
overall probability of success of the series system is given by
product of probability of success of individual components in series.

A failed component in series with any other component leads to failure
of the latter also.

\subsection{Parallel System} \label{parallel}

An $n$-path parallel system is one in which $n \geq 2$ parallel paths
$R_j$ are involved, where the system fails only if all the parallel
paths $R_j$ fail to work.

Let \(P(R_1), P(R_2),\ldots, P(R_n)\) denote the probabilities of
success of parallel paths \(R_1,R_2,\ldots,R_n\).  As the two possible
outcomes of any parallel path are a successful working of the path or
its failure, the probability of failure of paths
\(R_1,R_2,\ldots,R_n\) is given by
\(1-P(R_1),1-P(R_2),\ldots,1-P(R_n)\) respectively.

\subsubsection{Two-Path Parallel System}

We first discuss systems which contain exactly two paths in parallel.
For such systems the choice is represented using operators $\oplus$
and $\sqcup_{\psi}$.

If a process can be executed by either of two parallel paths $R_1$ or
$R_2$, that is, $R_1$ and $R_2$ are in parallel and both are equally
likely in terms of their selection, then this choice is represented as
$R_1 \oplus R_2$ and called non-probabilistic choice.  This is similar
to the operator used in reliability block diagrams to represent a
parallel system.

If a process can be executed by either of two parallel paths $R_1$ or
$R_2$, that is, $R_1$ and $R_2$ are in parallel and both are chosen
with different probabilities of selection, then this choice is
represented as $R_1\sqcup_{\psi}R_2$ and called binary probabilistic
choice, as in Andova~\cite{andova}.

\subsubsection{$n$-Path Parallel System}

We now discuss systems which contain $n$ parallel paths, where $n\geq
2$.  For such systems the choice is represented using
$\amalg_{\psi}$.

If a process in a system can be executed by any of the $n$ parallel
paths \(R_k$, $\forall k \in (1, 2, \ldots, n)\), where the probability
of selection of path $R_k$ is given by $\psi_k$ and
$\sum_{k=1}^{n}\psi_k=1$, then for such a system the choice is
represented using $\amalg_{\psi}$ as:
\begin{eqnarray} \label{eq123}
R_1 \amalg_{\psi_1} R_2 \amalg_{\psi_2} \ldots \amalg_{\psi_{n-1}} R_n.  
\end{eqnarray}

Using Theorem~\ref{th1}, the probability of success of the parallel
system, which we may denote by
\(P(R_1\amalg_{\psi_1}R_2\amalg_{\psi_2}\ldots\amalg_{\psi_{n-1}}R_n)\),
is given by:
\begin{eqnarray} \label{eq11}
\psi_1 P(R_1) + \psi_2 P(R_2)+ \ldots+ (1-\sum_{k=1}^{n-1}\psi_k) P(R_n).
\end{eqnarray}

As the two possible outcomes of a system are success and failure, the
probability of failure of the parallel system is given by:
\begin{equation*}
P_{\delta}(R_1 \amalg_{\psi_1}R_2 \amalg_{\psi_2} \ldots \amalg_{\psi_{n-1}}R_n) = 1-P(R_1\amalg_{\psi_1} R_2 \amalg_{\psi_2} \ldots \amalg_{\psi_{n-1}} R_n)
\end{equation*}

This leads to:
\begin{equation} \label{eq22}
P_{\delta} (R_1 \amalg_{\psi_1}R_2 \amalg_{\psi_2} \ldots \amalg_{\psi_{n-1}}R_n) = 1-\sum_{k=1}^{n}\psi_k P(R_k),
\end{equation}
\[\mathrm{where} \ \psi_n = 1 - \sum_{k=1}^{n-1} \psi_k.\]

Using Corollary~\ref{th2}, we can arrive at an analogous result with
$X_i$ instead of $R_k$, when each parallel path has exactly one
component.

\subsection{Assigning The Probability Of Selection Of Parallel Path $\psi_k$} \label{quantify}

We now consider a specific method (but obviously not the only one) to
assign the probability of selection $\psi_k$ for a parallel path
$R_k$, under the assumption that the probability of selection of any
path takes into account the probabilities of success and failure of
that path.

For any path $R_k$, we may have
\begin{eqnarray} \label{1}
\psi_k\propto\frac{P(R_k)}{P_{\delta}(R_k)},
\end{eqnarray}

and of course also,
\begin{eqnarray} \label{111}
\sum_{k=1}^{n}\psi_k=1{}. 
\end{eqnarray}

Between~\eqref{1} and~\eqref{111}, we can calculate and assign values to
the $\psi_k$.

\section{Probabilistic Approach To Reliability} \label{PATFT}

It is well known that considerations of probability distributions in
reliability analyses allow us to consider time-dependent functioning
of systems.  The probability of success is assumed to vary with time
according to some distribution.  This is a reasonable assumption, as
the probability of success or failure of any system depends on the
time for which it has been used.  We assume that all the
differently-aged components start working (or attempts are made to
start them working) at the same time.

As the exponential distribution is the one most commonly applied in
reliability analyses~\cite{book1,book2}, and is thus the basis of a lot
of existing theory, we show how to use our approach, with its
probabilistic parallel choice operators, to derive expressions for
system reliability parameters, assuming this distribution.  We note
that similar analyses are possible with other distributions as well.

\subsection{Exponential Distribution} \label{exp}

A general exponential distribution is given by \(\lambda_{i}
\exp(-\lambda_{i}(t-\hat t_{0i}))\).  Unlike other models, e.g., in Yu
and Fuh~\cite{Yu}, in our analysis we consider $\hat t_{0i}$ as the
time of installation of the component $X_i$, i.e., that different
components of the system may have different times of installation.

The probability of success of component $X_i$ is considered to vary
with time according to the exponential distribution, which of course
means that the probability of success of the overall system falls with
time, or that the probability of failure of the system increases with
time.  If $\lambda_{i}$ is the rate of failure of the component $X_i$,
the probability of success of component $X_i$ at and after time $T$ is
given by:
\begin{eqnarray*}
P(X_i,t \geq T)=\int_{T}^{\infty}\lambda_{i} \exp(-\lambda_{i}(t-\hat
t_{0i}))\ud t{}.
\end{eqnarray*}

This in turn gives:
\begin{eqnarray} \label{X}
P(X_i,t \geq T)=\exp(-\lambda_{i}(T-\hat t_{0i})).
\end{eqnarray}

\begin{enumerate}
\item \emph{Initially at time $t$ = }$\hat t_{0i}$
\[P(X_i,t)=\exp(-\lambda_{i}(0)) = 1.\]

This result is in line with our assumption that initially for time
close to the installation time, the probability of success of the
component is high and is almost one.

\item \emph{As time} \(t \rightarrow {t_{\infty}}\)
\[\lim_{t \to \infty} P(X_i,t \approx t_{\infty}) = \exp(-\lambda_{i}(\infty)) = 0.\]

This is in line with our assumption that at infinite time, the
probability of success of the component $X_i$ is very low, almost
zero.

\end{enumerate}

\subsection{Application Of This Distribution To Estimation Of System Parameters} \label{app}

\subsubsection{Reliability Time Estimation (RTE)}

Reliability time is the time until which the probability of success of
a system is greater than some minimum required probability of success
$\rho$.  This value $\rho$ is externally determined in advance
according to the desired system performance.

The formulae obtained are general and different from those obtained in
Bazovsky~\cite{book1} and Birolini~\cite{book2}, as our model itself is
more general than classical RBDs.

\begin{itemize}
\item[(i)] Generalized $n$-component series system:

As the probability of success of series system should be greater than
$\rho$, we get:
\begin{eqnarray*}
P(X_1\otimes X_2 \otimes X_3 \ldots X_n)\geq \rho
\end{eqnarray*}

Using~\eqref{eq1} and~\eqref{X}, we get:
\begin{equation} \label{snew1}
\rho \leq \exp\Big(-\lambda_{1}(t-\hat t_{01})\Big)\exp\Big(-\lambda_{2}(t-\hat t_{02})\Big) \ldots \exp\Big(-\lambda_{n}(t-\hat t_{0n})\Big)
\end{equation}

Taking the natural logarithm on both sides of
inequation~\eqref{snew1}, we have:
\begin{eqnarray*} \label{RTES}
t \leq -\frac{\ln (\rho)}{\sum_{i=1}^{n}\lambda_i} + \frac{\sum_{i=1}^{n}\lambda_i \hat t_{0i}}{\sum_{i=1}^{n} \lambda_i}.
\end{eqnarray*}

\item[(ii)] Generalized $n$-path parallel system:

We assume that each component $X_i$ is on a separate parallel path and
follows exponential distribution with installation time $t_{0i}$.

Since the probability of success of the parallel system should be
greater than $\rho$, we have:
\begin{eqnarray*}
\rho &\leq P(X_1)\psi_1 + P(X_2)\psi_2+ \ldots + P(X_n) (1-\sum_{i=1}^{n-1}(\psi_{i}).
\end{eqnarray*}

This in turn yields,
\begin{eqnarray} \label{snew5}
\rho \leq \sum_{i=1}^{n}\psi_i \exp\Big(-\lambda_i (t-\hat t_{0i})\Big).
\end{eqnarray}

The inequation~\eqref{snew5} can be solved for the following special
case.  (We are not aware of a closed-form general solution, but see
below for an approximation.)

Assume that all components have identical failure rates and
installation times:
\begin{enumerate}
\item\(\lambda_1 = \lambda_2 = \ldots = \lambda_n = \lambda\); and
\smallskip
\item\(\hat t_{01} = \hat t_{02} = \ldots = \hat t_{0n} = T_{0}\).
\end{enumerate}

This gives:
\begin{equation} \label{snew2}
\rho\leq\exp\Big(-\lambda(t-T_{0})\Big).
\end{equation}

Taking the natural logarithm on both sides of
inequation~\eqref{snew2}, we have:
\begin{equation}
t\leq T_{0}-\frac{\ln{\rho}}{\lambda}.
\end{equation}

Alternatively, we can attempt the general case by assuming a quadratic
approximation for the exponential function, provided the values of
$\lambda_i$ are small enough that cubic and higher order terms can be
disregarded.  This is reasonable as each $\lambda_i$ is supposed to be
very small.  Under this condition, for practically all values of time,
the exponential function can be approximated by a quadratic
expression.

Using~\eqref{eq11} and~\eqref{X}, and since the probability of success
of the system should be greater than $\rho$, we have:
\begin{equation*}
\rho\leq \psi_1 \exp\Big(-\lambda_1 (t-\hat t_{01})\Big) + \psi_2 \exp\Big(-\lambda_2 (t-\hat t_{02})\Big) + \ldots \\ + \Big(1-\sum_{i=1}^{n-1}\exp\Big(-\lambda_i (t-\hat t_{0i})\Big)\Big).
\end{equation*}

Therefore,
\begin{eqnarray} \label{snew3}
\rho \leq &\psi_1 \Big(1-\lambda_1(t-\hat t_{01})+\frac{\lambda_1^2 (t-\hat t_{01})^2}{2!} + \ldots\Big) \nonumber \\& + \ \psi_2 \Big(1-\lambda_2(t-\hat t_{02}) + \frac{\lambda_2^{2} (t-\hat t_{02})^2}{2!} + \ldots\Big) + \ldots \nonumber \\& + \ \Big(1-\sum_{i=1}^{n-1}\psi_i\Big)\Big(1-\lambda_n(t-\hat t_{0n}) +\frac{\lambda_n^{2}(t-\hat t_{0n})^2}{2!}+\ldots\Big),
\end{eqnarray}

taking as before, $\psi_n=(1-\sum_{i=1}^{n-1}\psi_i )$.

The inequality~\eqref{snew3} can be evaluated as:
\begin{eqnarray} \label{snew4}
&\rho\leq \Big(\frac{\sum_{i=1}^{n}(\psi_i (\lambda_i)^2)}{2}\Big)t^{2}-\Big(\sum_{i=1}^{n}\psi_i \lambda_{i}^{2} \hat t_{0i}+\sum_{i=1}^{n}\psi_i \lambda_i\Big)t \nonumber\\&+ \frac{\sum_{i=1}^{n}(\psi_i (\lambda_i)^2 (\hat t_{0i})^2)}{2}+\sum_{i=1}^{n}\psi_i \lambda_i \hat t_{0i} + 1
\end{eqnarray}

By solving the quadratic inequation~\eqref{snew4} and finding roots,
we can compute the region of time where the system is considered
acceptable.

Let $Q$ denote the determinant of the quadratic inequation.
\begin{eqnarray*}
Q&=\Big(\sum_{i=1}^{n}\psi_i \lambda_{i}^{2} \hat t_{0i}+\sum_{i=1}^{n} \psi_i \lambda_i\Big)^2 \\& - 2\Big(\sum_{i=1}^{n}(\lambda_i)^2 \psi_i\Big) \Big(1-\rho+\sum_{i=1}^{n}\psi_i \lambda_i \hat t_{0i}+\frac{\sum_{i=1}^{n}\psi_i (\lambda_i)^2 (\hat t_{0i})^2)}{2}\Big)
\end{eqnarray*}

It is evident that $Q \geq 0$ for all real values of $\psi_i$,
$\lambda_i$, $\hat t_{0i}$ $\forall i\in(1, 2, \ldots, n)$, and
$\rho$.

The roots of the inequation are thus real and are given by:
\begin{eqnarray*}
t_1=\frac{\sum_{i=1}^{n}\psi_i \lambda_{i}^{2} \hat t_{0i}+ \sum_{i=1}^{n}\lambda_i \psi_i + \sqrt{Q}}{2\sum_{i=1}^{n}\psi_i\lambda_i^2}
\end{eqnarray*}

\begin{eqnarray*}
t_2=\frac{\sum_{i=1}^{n}\psi_i \lambda_{i}^{2} \hat t_{0i}+ \sum_{i=1}^{n}\lambda_i \psi_i - \sqrt{Q}}{2\sum_{i=1}^{n}\psi_i\lambda_i^2}
\end{eqnarray*}

Using the obtained value of the roots, the temporal region where the
system is acceptable can be obtained.

The Newton-Raphson method and other numerical analysis methods can
also be used, and computer simulation can be used to calculate close
results.

\end{itemize}

\subsubsection{Mean Time To Failure (MTTF)}

The MTTF of a machine is the average time in which the system may
cease to work.

Given a system $S$, the MTTF is given by:
\begin{eqnarray}\label{in}
\mathrm{MTTF} = \int_{0}^{\infty}P(S,t\geq T)\ud T
\end{eqnarray}

The formulae obtained are generalized and the one for parallel systems
is different from that obtained using RBDs by Bazovsky~\cite{book1} and
Birolini~\cite{book2}, as our model is more general.

\begin{itemize}

\item[(i)] Generalized $n$-component series system

Without loss of generality, we assume that the installation time for
each component is the same, given by $T_{0}$.

For a $n$-component series system,~\eqref{eq1} holds.  Using~\eqref{X}
and~\eqref{in} and separating variables and constants, we get:
\begin{eqnarray} \label{snew6}
\mathrm{MTTF} = \int_{0}^{\infty} \exp\left(-\Big(\sum_{i=1}^{n}\lambda_i\Big)T\right) \exp\left(\Big(\sum_{i=1}^{n}\lambda_i\Big)T_{0}\right)\ud T.
\end{eqnarray}

Since for all components the installation time is $T_{0}$ and the
function is valid for time \(T \geq T_{0}\), here the second term
\(\exp\left(\sum_{i=1}^{n}\Big(\lambda_i\Big)T_{0}\right)\) is a
constant.  Therefore, ~\eqref{snew6} can be simplified to:
\begin{equation}\label{mttf1}
\mathrm{MTTF} = \frac{1}{\sum_{i=1}^{n} \lambda_i}
\end{equation}

\item[(ii)] Generalized $n$-path parallel system

We assume, as previously, that each parallel path has exactly one
component $X_i$, and as before that each component follows exponential
distribution and has installation time $t_{0i}$, $\forall i \in (1, 2,
\ldots, n)$.
 
For a $n$-path parallel system,~\eqref{eq11} holds.

Using~\eqref{eq11},~\eqref{X}, and~\eqref{in}, we get: 
\begin{eqnarray*}
\mathrm{MTTF} & = & \int_{0}^{\infty} \Bigg( P\Big(X_1,t\geq T\Big) \psi_1 + P\Big(X_2,t\geq T\Big) \psi_2 +\ldots \\ & & + \, P\Big(X_n,t\geq T\Big)\Big(1-\sum_{k=1}^{n-1}\psi_{k}\Big) \Bigg) \ud T.
\end{eqnarray*}

Since component $X_i$ has installation time $\hat t_{0i}$ and the
function is valid for time $t\geq \hat t_{0i}$, we have:
\begin{equation}\label{mttf2}
\mathrm{MTTF} = \sum_{i=1}^{n}\frac{\psi_i }{\lambda_i}. 
\end{equation}

\end{itemize}

\subsubsection{Probability Density Function (pdf)}

In probability theory, a pdf, or density of a continuous random
variable, is a function that describes the relative likelihood for
this random variable to occur at a given point.  The probability for
the random variable to fall within a particular region is given by the
integral of this variable's density over the region.

Thus the pdf is given by differential of the cumulative distribution
function.

For the exponential distribution, 
\begin{eqnarray} \label{pdf}
f(T)=-\frac{d}{\ud T}P(S,t\geq T),
\end{eqnarray}

which means:
\begin{eqnarray}
f(T)=\lambda\exp(-\lambda T).
\end{eqnarray}

The formulae obtained for this are general and different from those
obtained using RBDs in Bazovsky~\cite{book1} and Birolini~\cite{book2}.

\begin{itemize}

\item[(i)] Generalized $n$-component series system

For an $n$-component series system $S$,~\eqref{eq1} holds.

With $S$ given by \(S = X_1\otimes X_2\otimes X_3 \ldots \otimes
X_n\), where $X_1, X_2, \ldots X_n$ are components in series,
using~\eqref{pdf} and~\eqref{eq1}, we get:
\begin{equation}
f(T) = \Big(\sum_{i=1}^{n}\lambda_i\Big)\exp\Big(-\sum_{i=1}^{n} \lambda_i(T - \hat t_{0i})\Big).
\end{equation}

\item[(ii)] Generalized $n$-path parallel system

We assume, as previously, that each parallel path has exactly one
component $X_i$, which follows exponential distribution, and has
installation time $t_{0i}$, $\forall i \in (1, 2, \ldots, n)$.

For an $n$-path parallel system,~\eqref{eq11} holds.

With $S$ given by \(X_1\amalg_{\psi_1}X_2\amalg_{\psi_2}X_3\ldots
\amalg_{\psi_{n-1}}X_n\), using~\eqref{pdf} and Corollary~\ref{th2},
we get:
\begin{equation} \label{snew7}
f(T)=\sum_{i=1}^{n}\psi_i\lambda_i\exp\Big(-\lambda_i(T-\hat t_{0i})\Big).
\end{equation}

\end{itemize}

\subsubsection{Mean Time Between Failures (MTBF)}

Mean time between failures (MTBF) may be intuitively understood as the
expected time between two successive failures of a system.  It is
given by the expectation of the pdf.
\begin{eqnarray} \label{MTBF}
\mathrm{MTBF}=\int_{0}^{\infty}t f(t)\ud t
\end{eqnarray}

\begin{itemize}

\item[(i)] Generalized $n$-component series system

For an $n$-component series system,~\eqref{eq1} holds.

Using~\eqref{MTBF}, assuming without loss of generality that $T_{0}$
is the common installation time for all components $X_i$ in the series
system, we have:
\begin{equation*}
\mathrm{MTBF}=\int_{T_{0}}^{\infty}\Big(\sum_{i=1}^{n}\lambda_i\Big)\exp\Big(-\sum_{i=1}^{n}\lambda_i(t - T_{0})\Big) t \ud t,
\end{equation*}

keeping in mind that the function is valid for time $t\geq T_{0}$.
This reduces to:
\begin{equation}\label{mtbf1}
\mathrm{MTBF} = T_{0}+\frac{1}{\sum_{i=1}^{n}\lambda_i}.
\end{equation}

\item[(ii)] Generalized $n$-path parallel system

We assume, as previously, that each parallel path has exactly one
component $X_i$, and as before that each component follows exponential
distribution and has installation time $t_{0i}$, $\forall i \in (1, 2,
\ldots, n)$.

For an $n$-path parallel system,~\eqref{eq11} holds.  Using~\eqref{in}
and Corollary~\ref{th2}, we get:
\begin{eqnarray*}
\mathrm{MTBF}=\int_{0}^{\infty} \sum_{i=1}^{n}\psi_i \lambda_i\exp(-\lambda_i(t-\hat t_{0i}))t \ud t.
\end{eqnarray*}

Now each $X_i$ has a pdf valid for \(t \geq
\hat t_{0i}\), so
\begin{equation}\label{mtbf2}
\mathrm{MTBF} = \sum_{i=1}^{n} \psi_i\Big(\hat t_{0i}+ \frac{1}{\lambda_i}\Big)
\end{equation}

\end{itemize}

\subsubsection{Mean Time To Repair (MTTR)}

The mean time to repair (MTTR) of a system refers to the average time
required to repair a component that has failed or stopped working.  As
is common in prior literature, our analysis here excludes any
consideration of extraneous delays (e.g., due to lack of availability
of needed parts or equipment, or significant lead times to repair or
install components).

The MTTR of a system is given by the difference of mean time between
failures and mean time to failure:
\begin{eqnarray}\label{mttr}
\mathrm{MTTR} = \mathrm{MTBF} - \mathrm{MTTF}
\end{eqnarray}

\begin{itemize}

\item[(i)] Generalized $n$-component series system

For an $n$-component series system, using ~\eqref{mttf1},~\eqref{mtbf1},
and~\eqref{mttr}, we get:
\begin{eqnarray*}
\mathrm{MTTR} = T_{0}+\frac{1}{\sum_{i=1}^{n}\lambda_i}-\frac{1}{\sum_{i=1}^{n} \lambda_i}
\end{eqnarray*}
Thus the mean time to repair is given by:
\begin{eqnarray*}
\mathrm{MTTR} = T_{0}.
\end{eqnarray*}

\item[(ii)] Generalized $n$-path parallel system

For an $n$-path parallel system, using~\eqref{mttf2},~\eqref{mtbf2},
and~\eqref{mttr}, we get:
\begin{eqnarray*}
\mathrm{MTTR} = \sum_{i=1}^{n} \psi_i\left(\hat t_{0i}+ \frac{1}{\lambda_i}\right) - \sum_{i=1}^{n}\frac{\psi_i }{\lambda_i}
\end{eqnarray*}
Thus, the MTTR is given by:
\begin{eqnarray} \label{eqmttr}
\mathrm{MTTR} = \sum_{i=1}^{n}\psi_i\hat t_{0i}.
\end{eqnarray}

\end{itemize}

\subsubsection{System Failure Rate (SFR) ($\lambda_{eq}$)}

The failure rate for a single $X_i$ is $\lambda_i$ for the exponential
distribution.  We denote the system-wide equivalent failure rate as
$\lambda_{eq}$, also with the exponential distribution.  For
\(\lambda_{eq} > 0\), the probability of success of the system decays
exponentially with time \(t \geq \hat t_{0i}\).  A larger value of
$\lambda_{eq}$ implies that the decay in success probability will be
greater in a given amount of time, compared to a system with a smaller
value.

The formulae obtained are more general and different from those
obtained using RBDs in Bazovsky~\cite{book1} and Birolini~\cite{book2},
again as a consequence of our more general model.

For an individual component $X_i$, the failure rate ($\lambda_i$) for
the component can be obtained as:
\begin{eqnarray} \label{lam}
\lambda_{i}=\frac{\frac{-d}{\ud T} P(X_i,t\geq T)}{P(X_i,t\geq T)}.
\end{eqnarray}

\begin{itemize}

\item[(i)] Generalized $n$-component series system

For an $n$-component series system $S$ of components
$X_i$,~\eqref{eq1} holds, which gives:
\begin{equation} \label{snew9}
P(S,t\geq T) = \exp\Big(-\sum_{i=1}^{n}(\lambda_i)T\Big) \exp\Big(\sum_{i=1}^{n}(\lambda_i \hat t_{0i})\Big).
\end{equation}

For the system $S$, evaluating \(\frac{-d}{\ud T}P(S,t\geq T)\), we get:
\begin{equation}
\frac{-d}{\ud T} P(S,t\geq T) = \exp\Big(\sum_{i=1}^{n}(\lambda_i \hat t_{0i}\Big) - \left(\Big(\sum_{i=1}^{n}\lambda_i)T\Big) \Big(\sum_{i=1}^{n}\lambda_i\Big)\right).
\end{equation}

Using~\eqref{lam}, we get:
\begin{eqnarray*}
\lambda_{eq}=\frac{\frac{-d}{\ud T} P(S,t\geq T)}{P(S,t\geq T)}.
\end{eqnarray*}

This finally gives:
\begin{eqnarray} \label{eqy}
\lambda_{eq}= \sum_{i=1}^{n} \lambda_i.
\end{eqnarray} 

\item[(ii)] Generalized $n$-path parallel system

We assume, as previously, that each parallel path has exactly one
component $X_i$ that follows the exponential distribution and has
installation time $t_{0i}$.

For an $n$-path parallel system,~\eqref{eq11} holds.

For a system $S$ given by
\begin{equation}
S=X_1 \amalg_{\psi_1} X_2 \amalg_{\psi_2} X_3 \ldots \amalg_{\psi_{n-1}} X_n,
\end{equation}

we get,

\begin{equation}
\frac{-d}{\ud T} P(S,t\geq T) = \sum_{i=1}^{n-1} \lambda_i\psi_i\exp\Big(-\lambda_{i}(T-\hat t_{0i})\Big).
\end{equation}

We also have
\begin{eqnarray*}
\lambda_{eq}=\frac{\frac{-d }{\ud T}P(S,t\geq T)}{P(S,t\geq T)},
\end{eqnarray*}

which gives:
\begin{eqnarray} \label{eqsfr}
\lambda_{eq}=\frac{\sum_{i=1}^{n} (m_{i}\lambda_i)}{\sum_{i=1}^{n} (m_i)},
\end{eqnarray}

where \(m_i=\psi_i\exp(-\lambda_{i}(T-\hat t_{0i})), \ \forall i \in {1,
  2, \ldots, n}\).
  
The formula~\ref{eqsfr} indicates that the overall system failure rate
is a weighted mean in a manner of speaking, where the weights of
individual values $\lambda_i$ are given by $m_i$.

\end{itemize}

\subsection{Example} \label{eg}

Many data center systems used for online hosting utilize power supply
sources in parallel redundant design~\cite{Turner2008,Wiboonrat2008}.
The primary power supply comes from utility power.  Sometimes,
multiple power feeds to the facility are provided and are connected in
parallel originating from independent power grids.  A source of
secondary power is a UPS battery system in the power room, and a third
source of power is often a captive diesel generator, which may be
capable of producing tens or hundreds of kilowatts of power for longer
durations~\cite{Curtis2007}.

The three sources of power---utility power, UPS battery backup, and
diesel generator set---form three dissimilar parallel paths in the
data center power system.  The standard definition of parallel systems
fails in this model as the three paths are not equally likely and have
quite different characteristics.  While the utility power path is most
likely to be chosen, the costliest path involving diesel generators is
least likely to be taken, i.e., the path involving utility power is
given highest priority of being chosen while the UPS battery path is
given lower priority and diesel power path is given the least priority
of being chosen~\cite{Curtis2007}.

Thus associating a non-probabilistic choice with such parallel paths
would lead to erroneous results, and it is indeed known to be
problematic to estimate MTBF and such parameters for data center power
supplies~\cite{Torell2004}.  Thus probabilistic parallelism comes into
play.
 
For such a power supply system comprising utility power, UPS, and
diesel generator, as described above, the probability of success can
be given on the basis of Corollary~\ref{th2} and~\eqref{eq11} as:
\begin{eqnarray*}
P\left(X_1 \sqcup_{\psi_1} X_2 \sqcup_{\psi_2}X_3\right) =
 \psi_1 P(X_1)+\psi_2P(X_2) \, + 
 (1 - (\psi_1+\psi_2)) P(X_3),
\end{eqnarray*}
where $X_1$ represents the utility power component, $X_2$
represents the UPS battery, and $X_3$ represents the diesel generator
path, such that \(\psi_1>\psi_2>\psi_3\).

It is a simple matter, given data about such a system, to calculate
the degradation parameters: the MTBF using~\eqref{mtbf2}, the MTTR
using~\eqref{eqmttr}, the system failure rate using~\eqref{eqsfr},
etc.

\section{Conclusion}

In this paper, we have proposed an improved model to develop new
formulae for degradation analysis parameters of repairable or
maintainable probabilistic-choice systems.  These parameters include
MTTF, SFR, MTBF, MTTR, and pdf.

In the study of these parameters, each component in the system is
assumed to follow an exponential distribution.  The usage of
probabilistic choice operators in the selection of parallel paths is
motivated by the need to address more real world situations like
sophisticated power supply systems for data centers, which cannot be
modeled using RBDs and other classical approaches.  The license for
considering different installation times for components adds more
flexibility to the model, which thus too better fits real-world
systems where different components in a system may have different
installation times.

Finally, we point out one important direction for future research.
The model proposed can be extended in a straight-forward way to other
distributions like Weibull, normal, etc., to develop formulae for
parameters related to degradation analysis, along the lines we
indicate in our approach, which assumed that components follow
exponential distributions.  This will help in analyzing systems which
are composed of components which follow such different distributions.

\end{document}